\documentclass[11pt]{article}

\usepackage{SIGACTNews}

\usepackage{amsthm}

\newtheorem{lemma}{Lemma}
\newtheorem{definition}{Definition}
\newtheorem{theorem}{Theorem}
\newtheorem{corollary}{Corollary}
\newtheorem{proposition}{Proposition}
\newtheorem{property}{Property}
\newtheorem{obs}{Obsevation}

\usepackage{amssymb}
\usepackage[pdftex]{graphicx}
\usepackage{ifpdf}

\usepackage{amsmath}
\usepackage{array}

\usepackage{stfloats}
\usepackage{fixltx2e}

\usepackage{algorithm}
\usepackage[noend]{algpseudocode}

\date{}

\title{SPAIDS and OAMS Models in Wireless Ad Hoc Networks}
\author{Aikaterini Nikolidaki \footnote{School of Electrical and Computer Engineering, National Technical University of Athens, Email: aiknikol@mail.ntua.gr}}

\begin{document}
\SIGACTmaketitle


\begin{abstract}
	In this paper, we present two randomized distributed algorithms in wireless ad hoc networks. We consider that the network is structured into pairs of nodes (sender, receiver) in a decay space. We take into account the following: Each node has its own power assignment and the distance between them does not follow the symmetry property. Then, we consider a non-uniform network or a realistic wireless network, which is beyond the geometry. Our model is based on the Signal to Interference plus Noise Ratio (SINR) model. In this work, the main problem is to solve the scheduling task aiming the successful transmission of messages in a realistic environment. Therefore, we propose the first randomized scheduling and power selection algorithm in a decay space and is called as SPAIDS. In order to solve the problem in this non-uniform network, we introduce a new way to study the affectance (the interference) among the links, which is defined as Weighted Average Affectance (WAFF). Moreover, we study the online broadcast problem in a metric space, in which the nodes are activated in case that they receive packets. We propose an online algorithm in a metric space which is denoted as OAMS. Our aim is to obtain the maximum subset of nodes that receive the message from a sender node with enough energy supplies. Finally, we compare the performance of OAMS to the optimal.
\end{abstract}

\section{Introduction}
In wireless networks, a great challenge is the management of simultaneous transmissions among nodes in an environment which is characterized by real conditions. We are concentrated on the \emph{scheduling problem}, where the nodes are located in an \emph{arbitrary decay space}. In this space, the transmission signal may be reduced by the interference of other communication links, the obstacles, the reflections and the shadowing. Thus, we consider two conditions: Firstly, the distances between nodes are not symmetrical. Secondly, each of the nodes has not the same transmission power. These two conditions add a greater degree of difficulty in our study. Therefore, our aim is to seek the fewest number of different time slots needed to schedule all the communication links in such a network. The key point of our study is to ensure the successful transmission of messages in a decay space. Then, we use acknowledgement messages and determine guards to protect the transmissions providing quality of services. 

Also in this paper, we are focused on the broadcast problem, where a sender node transmits messages to all the nodes in the network when they follow the symmetry property. In particular, we present an online mechanism in a metric space that each receiver node is activated to get a message from its sender. We consider that each mobile user has a \emph{limited battery capacity} or equivalently \emph{battery-feasibility}. In \cite{badanidiyuru2012learning} and \cite{singer2010budget}, the authors study online mechanism and mainly under budget of an user (agent) in an online procurement market.

We adopt the \emph{Signal-to-Interferences-and-Noise-Ratio} (SINR) \emph{physical model} which is based on the physical assumptions that the strength of signals reduces gradually because of the cumulative interference of other communication links. The SINR model recently acquires the attracted study of algorithmic community. Moscibroda and Wattenhofer \cite{moscibroda2006complexity} initiated a scheduling algorithm in the SINR model in which a set of links is successfully scheduled into polylogarithmic number of slots. In \cite{kesselheim2010distributed}, randomized distributed algorithms were proposed for the scheduling problem, where a transmission probability is used, as a parameter which works for short schedules. Jurdzinski et al. \cite{Jurdzinski:2014} presented a randomized algorithm, in which all nodes start the algorithm at the same time, and a randomized algorithm, in which the source node is only actived during the initiation phase of the algorithm. They studied this problem in an uniform network using a communication graph in a metric space with a distance function at most 1. Bodlaender and Halld\'{o}rsson \cite{bodlaender2014beyond} used an abstract SINR model in order to solve the capacity problem with uniform power in a decay space. In \cite{yu2016distributed}, the authors presented a randomized multiple-message broadcast protocol.
\subsection{Additional Related Work}
The study of scheduling (and capacity) problem in the SINR model through algorithmic analysis using oblivious power schemes presented recently in the literature. In these schemes, the power chosen for a link depends only on the link length itself and these can be categorized into three cases: 1) the uniform power; 2) the linear power; and 3) the mean power scheme. The first $O(\log\log\Delta)$-approximation algorithm for oblivious power schemes is presented in \cite{HalldorssonT15,HT} for the \emph{wireless scheduling problem} and the \emph{weighted capacity problem}, where $\Delta$ is the ratio of the maximum and minimum link lengths. The result is achieved by the representing of interference by a conflict graph. However, the unweighted Capacity problem admits constant-factor approximation according to \cite{Kesselheim:2011:CAW:2133036.2133156}. The WCapacity problem admits $O(\log^{*} \Delta)$-approximation according to \cite{HT,HalldorssonT15}. The result for the scheduling problem is $O(\log n)$-approximation. In the case of the grouping of link lengths, the $O(\log \Delta)$-approximation, is according to \cite{fu2009power,Goussevskaia:2007:CGS:1288107.1288122,Halldorsson12}. A $O(\log^{*} \Delta)$-approximation for the scheduling problem is presented at \cite{HalldorssonT15,HT}, which is the best bound.

Note that in \cite{APX}, Halld\'{o}rsson and Wattenhofer prove that the wireless scheduling problem is in APX. More, in \cite{Halldorsson12}, the author present an approximation algorithm for the wireless scheduling problem with ratio $O(\log n \cdot \log\log \Delta)$. These results hold also for the weighted capacity problem. In \cite{halldorsson2013power}, an algorithm for the capacity problem that achieves $O(\log\log\Delta)-$ is proposed. In addition, T. Tonoyan \cite{TT} prove that a maximum feasible subset under mean power scheme is always within a constant factor of subsets feasible under linear or uniform power scheme for the capacity problem.
\subsection{Contribution}
Following \cite{Jurdzinski:2014}, we study the scheduling problem in a more general space, in an arbitrary decay space. This means that the strength of a transmitted signal of any sender node is vulnerable because of interference of other nodes, obstacles, reflections and shadowing. Therefore, there is a reduction of the strength signal and the receiver may not get the message from its sender. In this paper, we achieve the ideal solution of the scheduling problem using power control through two different efficient algorithms. The main contributions are summarized:

Firstly, we propose the first randomized distributed algorithm in order to control the power of each node and to solve the minimum scheduling problem in a non-uniform network. The algorithm is based on the coloring method in \cite{Jurdzinski:2014}, which assigns probability/color to each node taking part in an implementation. We propose a scheduling and power selection algorithm in a decay space, which is called as SPAIDS. Therefore, we propose an $O(\log^{*}\Delta \log n)$ randomized algorithm, where $n$ is the number of nodes and  $\Delta $ is the ratio between the maximum and the minimum power assignment. More details:
\begin{itemize}
	\item We determine a set of probability transmissions to each node in order to achieve transmissions of the messages in the network and we separate it into $K$ subsets. The SPAIDS algorithm needs a $O(\log^{*}\Delta)$ time in order to assign colors in the nodes because of the separation into feasible subsets. 
	\item We consider that the message is successfully received when the sender receives an acknowledgement message from its receiver because the nodes are located in a decay space.
	\item We use guards in order to protect the receiver from interference of other links and to boost the signal. Also, we protect the sending of an acknowledgement message from a receiver. Then, we guarantee the successful transmission and the quality of service.	
\end{itemize}

Secondly, we propose online algorithm in a metric space, which is called as OAMS, in order to control the power of each node and to achieve the deliver of messages to all nodes in the space. We focus on the online broadcast problem that each receiver node is activated at each time step. Each mobile user has a battery, who can store power at most $C_{B}$. Also, we assume that there is unknown distribution of nodes in our network. The algorithm assigns probability/color to each node taking part in an implementation. The proposed algorithm is constant-competitive.

\subsection{Paper Organization}
The rest of this paper is organized as follows: Section 2 describes the system model used in this work and gives some useful definitions. Section 3 presents a conflict graph and its properties in a decay space as well as upper bounds. Section 4 presents the scheduling and power selection algorithm (SPAIDS). Section 5 presents the online algorithm in a metric space (OAMS). 
\section{System Model and Definitions}
In this section, we describe the proposed model of wireless ad-hoc networks, which consists of pairs of nodes. A pair of nodes is denoted as a quasi-link $q_{i}=(s_{i}, r_{i})$, where $s_{i}$ is the sender and $r_{i}$ is the receiver of quasi-link $i$. We consider that quasi-links are the communication links in decay spaces (Section \ref{mdc}).  The model is characterized by the following components: SINR formula, decay signal among nodes using quasi-metrics and bounded growth properties. We study the case that the power transmission is non-uniform in all the nodes as well as the case that the distances among nodes are not symmetrical. Thus, our proposed model is characterized as a realistic model. Moreover, we introduce a new notion of affectance, the Weighted Average Affectance (WAFF).
\subsection{System Model}
We consider that a wireless network can be represented as a graph $G=(V,E)$, where $V$ is the set of nodes and $E$ is the set of edges (or quasi-links). Each directed edge $q_{i}$ is denoted as a communication request from a sender $s_{i}$ to a receiver $r_{i}$ in decay spaces. We consider $S = \{q_{1},...,q_{n}\}$ is a set of quasi-links. Each sender $i$ transmits packets to its receiver at power $P_{i}$ multiplied by the gain $G_{ij}$. The gain represents the distance between sender and receiver, which is denoted as $G_{ij}=1/d_{ij}^{a}$, where $a\in (2,6)$ is the path-loss exponent and $d_{ij}^{a}=d(s_{i}, r_{j})^{a}=q(s_{i}, r_{j})$ is the quasi-distance (in Section \ref{mdc}) among two nodes $i$ and $j$.

In our model, we use the SINR interference model and assume that $v_{i}$ is the noise (constant) at the receiver $i$ and $\beta_{i}$ is a threshold of SINR. The signal transmission can be successful if and only if $SINR\geq \beta_{i}$ for all the senders $i$, where $SINR_{i} = \dfrac{P_{i}/l_{i}^{a}}{\sum_{i\neq j} P_{j}/d_{ji}^{a} + v_{i}}\geq \beta_{i}$. 
\subsection{Metric and Decay Spaces}\label{mdc}
In \cite{HT,HalldorssonT15,Jurdzinski:2014}, the nodes of network are embedded in a general metric space. 
A metric space consists of an ordered pair of $(\mathcal{V},d)$, where $\mathcal{V}$ is a set of nodes and $d: \mathcal{V}\times \mathcal{V}\rightarrow \mathbb{R}_{+}$ is a distance function. $d$ is defined as a metric such that for any $u,v,w \in \mathcal{U}$, the following holds: (i) symmetry property, (ii) triangle property and (iii) non-negativity property \cite{Gupta}.

On the other hand, a real network has not the symmetry property. In \cite{bodlaender2014beyond}, the authors study their network in a metric space when there is not the symmetry property. This metric space is defined as a decay space or else quasi-metric. A quasi-metric on a set $V$ is defined as a function $q:V\times V \rightarrow \mathbb{R}_{\geq 0}$ such that for all $v,u,w\in V$: (i) $q(u,v)\geq 0$, (ii) $q(u,v)=q(v,u)=0 \Leftrightarrow u=v$ and (iii) $q(u,v)\leq q(u,w) + q(w,v)$ \cite{romaguera2000semi}. We denote the quasi-distance of two nodes $i,j$: $q(s_{i},r_{j})=d(s_{i},r_{j})^{a}$, where $a\in (2,6)$. Each quasi-link $i$ is $q_{i}=q(s_{i},r_{i})=d(s_{i},r_{i})^{a}$.

Moreover, we bound the arbitrary growth of space. The bounded growth decay space consists of two properties: (i) Doubling Dimension. This property is the infimum of all numbers $\delta >0 $ such that every ball of radius $r>0$ has at most $C \epsilon^{-\delta}$ points of mutual distance at least $\epsilon r$ where $C\geq 1$ is an absolute constant $\delta > 0$ and $0<\epsilon\leq 1$. Metrics with finite doubling dimensions are said to be doubling. (ii) Independent Dimension. In the decay spaces, the concept of independence-dimension $D$ is applied in \cite{bodlaender2014beyond,goussevskaia2009capacity} and is defined as follows: Let $\left( V,q\right) $ be a metric space and $v\in V$. A set $I\subseteq V\setminus{\{v\}}$ is called independent with respect to $v$ if $B(w, q(v,w))\cap I = \{w\}$ for all $w \in I$. The size of the largest independent point set is called the independent-dimension of $(V,q)$ and denoted by $D$.	
\subsection{Power Conditions}
In addition, we give two conditions for the power assignments: (i) $P_{v}\geq c\beta N q_{v}$ for some constant $c>1$. (ii) If $q_{v}^{-1}\leq q_{w}^{-1}$ then $P_{v}\leq P_{w}$ and $P_{v}\cdot q_{v}^{-1}\leq P_{w} \cdot q_{w}^{-1}$, that large quasi-link in a decay space has small power assignment. While small quasi-link has better power condition in order to transmit a message to the receiver in a decay space.
\subsection{Affectance}
In this part, we introduce a new notion of affectance. It is defined as a "Weighted Average Affectance (WAFF)" and depends on the quasi-link lengths, the power assignments and the density bounding properties. Note that the measure of affectance is introduced by \cite{APX} and recently reused by \cite{HT}. In \cite{HT,HalldorssonT15,APX}, the authors study only the quasi-link lengths. 

We study the model in decay spaces when the distances between of nodes are not symmetrical. Thus, we have $d(x,y)\neq d(y,x)$. Also, the nodes have not the same power assignment. The next definition means that if each quasi-link $q_{j}$ has weighted affectance on a quasi-link $q_{i}$ then we have the weighted average affectance $WAFF(S,i)$ of a set of quasi-links $S$ on a quasi-link $q_{i}$, where the nodes are located in a decay metric space. In this paper, the weight of a node $i$ is the transmitted probability $p_{i}$ of node $i$. For simplicity we redefined the $WAFF(S,i)$ as $Da_{p}(S,i)$. Then, we have:
\begin{definition}\label{ddap}
	Let $S$ be a set of quasi-links. We consider a quasi-link $i\notin S$. The "Weighted Average Affectance" of $S$ on $i$ in a decay space is defined as follows:
	\begin{equation}
	WAFF_{p}(S,i)=Da_{p}(S,i)=\dfrac{\sum_{j\in S} p_{j}\cdot a_{P}(j,i)}{\sum_{j \in S} p_{j}},
	\end{equation}
	where $a_{p}(j,i) = \max \left\lbrace \dfrac{R\cdot q_{i}}{q(s_{j}, r_{i})}, \dfrac{R\cdot q_{i}}{q(r_{j}, r_{i})}, \dfrac{R\cdot q_{i}}{q(s_{j}, s_{i})}, \dfrac{R\cdot q_{i}}{q(r_{j}, s_{i})} \right\rbrace$ is the affectance of quasi-link $j$ on quasi-link $i$ using power assignments $P_{i}$ and $P_{j}$ in a decay space, respectively. The decay distance is $q(\cdot,\cdot)$ and $R=P_{j}/P_{i}$.	 	
\end{definition} 
\section{Conflict Graph in Decay Space}
In this section, we study conflict graphs and their properties in non-uniform wireless networks. Note that conflict graphs are graphs defined over a set of links (quasi-links in decay spaces). Our interest is situated in the case of the non-symmetry property and the non-uniform power assignment. Useful definition is the independence of quasi-links, as it determines the less distance of quasi-links when they are not in conflict. Let  $q_{i}=(s_{i},r_{i})\notin S$ be the quasi-link in which $s_{i}$ sends a message to $r_{i}$. Our goal is to seek an upper bound of the WAFF of a set $S$ of quasi-links on the given quasi-link $q_{i}\notin S$. 

In this paper, our study is based on non-unit balls (of radius $r\neq 1$) because of quasi-links. We divide the set $S$ in annuli disks (or n-spheres in distance $D\geq 3$) centered at the $r_{i}$ (or $s_{i}$) of quasi-link $i$. In Figure \ref{fig:disks}, the concentric disks surrounded around the endpoint of quasi-link $i$ are represented. The set $S$ consists of equilength subsets $S_{k} \subseteq S$. Let $B_{k}(r_{i},d_{k})$ be the annulus disk with center $r_{i}$ and radius $d_{k}$ for each $k$ disk. Each $S_{k}$ has a number of active nodes, which can cause an affectance on the quasi-link $i$. Each $k$ disk has a number of annuli disks $B_{k\lambda}(v_{k\lambda},\rho)$ and $B_{k\lambda}(v_{k\lambda},a\rho)$, where $\rho$ is the radius at the small disk and $a\rho$ at the large disk. The large disk $B_{k\lambda}(v_{k\lambda},a\rho)$ is at most a factor of the radius $d_{k}$. 
\begin{figure}
	\centering
	\includegraphics[width=2in]{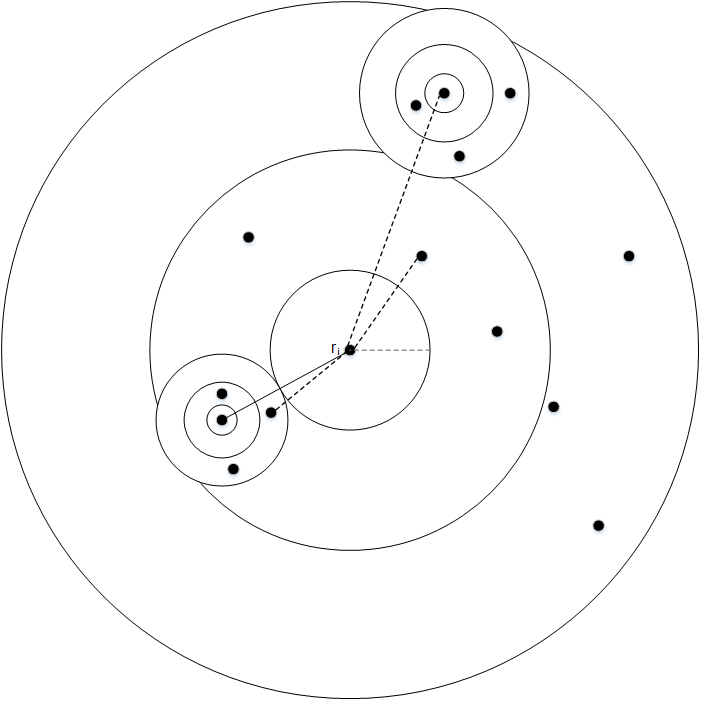}
	\caption{Concentric disks surrounded around the endpoint of quasi-link $i$.}
	\label{fig:disks}
\end{figure}
In the following part, we introduce the concepts of DP-feasible and acknowledgment messages. Also, we seek lower and upper bounds. 
\subsection{Feasibility}
A set $S$ of quasi-links is called as \emph{DP-feasible} if the $SINR$ holds for each quasi-link $i \in S$ in case that we use a power assignment $P$. The set $S$ is \emph{feasible} if there exists a power assignment $P$ for which $S$ is \emph{DP-feasible}. Thus, a set S of quasi-links is feasible if and only if the average weighted affectance satisfies: $Da_{p}(S,i) \leq 1/C_{2}\beta$. A set of quasi-links is called $\beta$\emph{-DP-feasible} if it is \emph{DP-feasible}. The Proposition \ref{propfeas} gives the feasibility of WAFF.
\begin{proposition}\label{propfeas}
	We assume $C_{1}$ and $C_{2}$ are the upper and the lower bound of the sum of probabilities of the transmitted nodes, correspondingly.	A set of quasi-links S is \emph{DP-feasible} if and only if $Da_{p}(S,i)\leq 1/\left( \beta\cdot C_{2}\right)$.
\end{proposition}
\begin{proof}
Initially, we use the WAFF and the fact that the affectance of a set of links $S$ on the quasi-link $i$ is upper bounded by the value $1/\beta$ according to \cite{HT}. In \cite{Jurdzinski:2014}, the authors use unit balls in uniform networks. In decay spaces, we consider the following: Let disks $k>1$. For each node $j$, there is a color $p$ such that the sum of probabilities of  this color, in ball $B' \equiv \bigcup_{k>1}\left(  B_{k}(i,Rq\gamma_{1}+\frac{k-1}{2}\cdot q)\setminus B_{k-1}(i,Rq\gamma_{1}+\frac{k-2}{2}\cdot q)\right)$, is at least $C_{2}$ whp: $\sum_{\substack{j: p_{j}=p \\	j\in B'}} p_{j} \geq C_{2}$, where $R$ is the ratio of power assignment $(\frac{P_{j}}{P_{i}})$, $\gamma_{1}>0$ is a constant; and $q$ is the decay distance. 

Then, $Da_{p}(S,i) =\dfrac{\sum_{j\in S} p_{j}\cdot a_{P}(j,i)}{\sum_{j \in S} p_{j}} \leq \dfrac{\sum_{j\in S} p_{j}}{\beta C_{2}} \leq \dfrac{1}{\beta C_{2}}$.
\end{proof}

\begin{proposition}\label{mp}
	Two quasi-links $i$,$j$ in a decay space with $q_{j}\geq q_{i}$ are $\gamma_{1}$-independent iff $q(j,i)> \gamma_{1} q_{i}$ and are $(R, \gamma_{1})$-independent iff $q(j,i)> \frac{1}{2}\cdot R \gamma_{1} q_{i}$ or $q(j,i) q(i,j)> \frac{1}{4}\cdot \gamma_{1}^{2} q_{i} q_{j}$ with probability of transmitting $p_{j}\geq 1/2$.  	
\end{proposition}
\begin{proof}
	Let S be a set of links. We consider power $P_{i}$ for each link i. A set S is DP-feasible iff $Da_{p}(S,i)= \sum\limits_{j\in S}Da_{p}(j,i)<1/\left( \beta\cdot C_{2}\right) $. From the feasibility of S holds the next inequality: $q(j,i)> \frac{P_{j}}{P_{i}}\cdot q_{i}\cdot C_{2}\cdot \beta $ and therefore the quasi-distance links $i,j$, where we denote $\gamma_{1}= C_{2}\cdot \beta$, is defined by $q(j,i)> \frac{P_{j}}{P_{i}} \cdot q_{i}\cdot \gamma_{1}$ or $q(j,i)>\frac{1}{2} \cdot R \cdot q_{i}\cdot \gamma_{1}$. Then, $q(j,i) q(i,j)> \frac{1}{4}\cdot \gamma_{1}^{2} q_{i} q_{j}$.
\end{proof}

In general decay space, we introduce a lower bound in the WAFF for each node $k\in B(j,q_{ji})$, as follows: A set of quasi-links S is \emph{GDP-feasible} if and only if $Da_{p}(S,i)\geq \delta_{1}/ C_{DI}$, where $0<\delta_{1}<1$ and $C_{DI}$ is an upper bound of the sum of probabilities of the transmitted nodes in independent dimension D.
\subsection{Acknowledgements}\label{acknow1}
In this part, we study the reverse case that a receiver $r_{i}$ transmits an acknowledgement message to its sender $s_{i}$ in order to be known that $r_{i}$ successfully received the message. The transmission of an acknowledgement message needs a proper power assignment. We take into account that the nodes are located in a decay space. We need to study the affectance of quasi-link $q_{j}=l_{j}^{a}$ by quasi-link $q_{i}=l_{i}^{a}$ in the case of acknowledgement transmissions, where $l_{i}$ and $l_{j}$ are links in uniform networks. Also, we compare the affectance of acknowledgments with the standard $Da_{p}$ affectance. 
\begin{definition}\label{defp} Let a quasi-link $q_{i}=(s_{i},r_{i})$ and its reverse quasi-link $q_{i}^{*}=(r_{i},s_{i})$ for the acknowledgement transmission. We consider that a quasi-link $i$ has a power assignment $P_{i}>0$. However, the acknowledgement transmission needs a power $P_{i}^{*}>0$ in order to arrive at the sender. We consider that the power $P_{i}^{*}$ is defined by $	P_{i}^{*}=\frac{P_{i}\cdot q_{i}^{*}\cdot q_{ji}}{q_{i}\cdot q_{ji^{*}}}=\frac{P_{i}\cdot {l_{i}^{*}}^{a}\cdot d_{ji}^{a}}{l_{i}^{a}\cdot d_{ji^{*}}^{a}}$.
\end{definition}
\begin{lemma}\label{acknowled}
	For all quasi-links of a set $S$, it holds that $a_{P^{*}}(q^{*}_{i},q^{*}_{j}) = O(a_{P}(q_{i},q_{j}))$ when the symmetry property is not satisfied and each node has its own power assignment $P$. 
\end{lemma}
\begin{proof}
	The affectance of acknowledgement of quasi-link $q_{j}$ by quasi-link $q_{i}=l_{i}^{a}$ in a decay space is: $a_{P^{*}}(i,j)= a_{P^{*}}(q_{i}^{*},q_{j}^{*})=\frac{P_{j}^{*}\cdot {l_{i}^{*}}^{a}}{P_{i}^{*}\cdot {d_{ji}^{*}}^{a}}$. By the Definition \ref{defp}, $a_{P^{*}}(l_{i}^{*},l_{j}^{*}) =\frac{P_{j}^{*}\cdot {l_{i}^{*}}^{a}\cdot l_{i}^{a}\cdot d_{ji^{*}}^{a}}{P_{i}\cdot {l_{i}^{*}}^{a}\cdot d_{ji}^{a}\cdot {d_{ji}^{*}}^{a}}\leq \frac{P_{j}\cdot l_{i}^{a}\cdot d_{ji^{*}}^{a}}{P_{i}\cdot d_{ji}^{a}\cdot {d_{ji}^{*}}^{a}}\leq \frac{P_{j}\cdot l_{i}^{a}}{P_{i}\cdot d_{ji}^{a}} = O(a_{P}(q_{i},q_{j}))$
	
	By the triangle inequality, it holds the last inequality:
	$d_{ji^{*}} = d(j, i^{*}) = d(s_{j},s_{i}) \leq d(s_{j},r_{i}) + d(r_{i},s_{i}) \leq l_{j} + l_{i} + d(r_{j},s_{i}) + l_{i}^{*}\leq (1+\frac{1}{R^{1-\epsilon\cdot \gamma_{1}}})\cdot d(r_{j},s_{i})\\
	= (1+\frac{1}{R^{1-\epsilon\cdot \gamma_{1}}})\cdot d(j^{*},i^{*}) = (1+\frac{1}{R^{1-\epsilon\cdot \gamma_{1}}})\cdot d_{ji}^{*}$. 
	
	However, we want: $\min \left\lbrace d_{ji^{*}}, d_{ji}^{*} \right\rbrace \leq \min \left\lbrace (1+\frac{1}{R^{1-\epsilon\cdot \gamma_{1}}})\cdot d_{ji}^{*}, d_{ji}^{*} \right\rbrace = d_{ji}^{*}$.
\end{proof}
\begin{lemma}
	For all quasi-links of a set $S$, it holds that $Da_{P_{*}}(q^{*}_{i},S^{*}) = O(Da_{P}(S,q_{i}))$ when the symmetry property is not satisfied.
\end{lemma}
\subsection{Upper Bound Graphs}	
A conflict graph for a set of quasi-links is an \emph{upper bound graph}, if each independent set in this conflict graph is $DP$-feasible using a power assignment. A conflict graph for a set of quasi-links is a \emph{lower bound graph}, if each $DP$-feasible set is an independent set in this conflict graph. Therefore, upper and lower bounds for the scheduling problem are sought. In \cite{HT,HalldorssonT15}, the authors introduce the initial idea to seek bounds and show that there are $O(\log\log \Delta)$-approximation algorithms for Scheduling and WCapacity using oblivious power schemes in a metric space. 

In the case of doubling dimension, we seek an upper bound of the WAFF of an set $S$ of quasi-links on a given quasi-link $q_{i}\notin S$. The set $S$ consists of equilength subsets $S_{k} \subseteq S$. Each $S_{k}$ has a number of nodes. The set $S$ is divided into two subsets $S'$ and $S''$. The subset $S'$ contains the quasi-links that are closer to sender of quasi-link $i$: $S'=\left\lbrace j \in S: D_{1} \geq D_{2}\right\rbrace$, where $D_{1}=\min \left\lbrace q(s_{j},r_{i}), q(r_{j},r_{i})\right\rbrace$ and $D_{2}=\min \left\lbrace q(s_{j},s_{i}), q(r_{j},s_{i})\right\rbrace$. The subset $S''$ contains the quasi-links that are closer to receiver of quasi-link $i$: $S''=\left\lbrace j \in S: D_{3} \leq D_{4}\right\rbrace$, where $D_{3}=\min \left\lbrace q(s_{j},r_{i}), q(r_{j},r_{i})\right\rbrace$ and $D_{4}=\min \left\lbrace q(s_{j},s_{i}), q(r_{j},s_{i})\right\rbrace$. 	
\begin{lemma}\label{ll2}
	Let $\gamma_{1}\geq 1$, $S$ be a set of 1-independent quasi-links. The quasi-links $i, j$ are $(R, \gamma_{1})$-independent, $\forall j \in S$. Then, the WAFF in a decay space is given as follows:
	\begin{equation}
	Da_{R}(S,i) \in \frac{C_{DI}}{C_{2}} O\left( \gamma_{1}^{m-2} \frac{q_{i}}{q} R^{m-1} (1 + (R\gamma)^{-1}) \right)
	\end{equation}
	where the transmission probability of nodes in each $k$ disk is bounded by $C_{DI}$ and $C_{2}$.	
\end{lemma}
\begin{proof}
	The affectance of quasi-link $i$ by quasi-link $j$ in a decay space is:
	\begin{equation}
	a_{p}(j,i)= \frac{P_{j}\cdot q_{i}}{P_{i}\cdot q(j,i)} = \frac{R \cdot q_{i}}{\left( R q \gamma_{1} + (k-2) q\right)}
	\end{equation}
	
	Then, the weighted average affectance for each $k$ disks is calculated as follows:
	\begin{equation}
	\begin{split}
	& \frac{1}{\sum_{j\in S_{k}\setminus S_{k-1}} p_{j}} \sum_{j\in S_{k}\setminus S_{k-1}}  p_{j} a_{p}(j,i)\leq \frac{1}{C_{2}} \sum_{j\in S_{k}\setminus S_{k-1}}  p_{j} a_{p}(j,i) \leq \frac{1}{C_{2}} \sum_{j\in S_{k}\setminus S_{k-1}} p_{j} \frac{R q_{i}}{R \gamma_{1} q + (k-2)q}\\ 
	& \leq \frac{C_{1}}{C_{2}} R \frac{q_{i}}{q} \sum_{k\geq 2} |S_{k}|\frac{1}{(R\gamma_{1} + k-1)^{2}} = \frac{C_{1}}{C_{2}} O\left( \gamma_{1}^{m-2} \frac{q_{i}}{q} R^{m-1} (1 + R\gamma) \right)\\
	\end{split}
	\end{equation}
	
\end{proof}

\begin{corollary}\label{mc1}
	Let L be a 1-independent set of links. The quasi-links $i,j$ s.t. $q_{i} \geq q_{j}$ and are $(R,\gamma_{1})$-independent, $\forall j \in L$. There exists $m>1$ and the path-loss $a>m$, the power of link i is greater than the power of link j, $P_{j}=O(P_{i})$. Then, the weighted average affectance in a decay space is given by $Da(L,i) = \frac{C_{1}}{C_{2}} O\left( \gamma_{1}^{m-2} \right)$,where the transmission probability of nodes in each $k$ disk is bounded by the parameters $C_{1}$ and $C_{2}$.	
\end{corollary}
\section{Scheduling and Power Selection Algorithm in a Decay Space}\label{spaids1}
In this section, we propose a randomized distributed algorithm in order to control the power of each node and to solve the minimum scheduling problem in a non-uniform network. \begin{algorithm}
	\caption{Scheduling and Power Selection Algorithm in a Decay Space (SPAIDS)}
	\begin{algorithmic}
		\State $\textbf{Initialization:}$ We consider a set of quasi-links $q_{1},...q_{n}$ and a set of probability transmissions $C_{1}/2n,...,p_{\max}$ and $K=\lceil2C_{1}\beta'/\beta C_{2} \rceil$. The initial set of feasible links $S=\emptyset$ and $0<\gamma_{1}<1$
		\State $\textbf{Procedure:}$
		\State $\textbf{for}$ $t=1$ to $T$ $\textbf{do}$
		\State $\hspace{ 2 mm} \textbf{for}$ $\mu = 0$ to $K-1$ $\textbf{do}$
		\State $\hspace{ 4 mm} \textbf{for}$ $p_{i}=\frac{C_{1}}{2n}+\mu(\frac{p_{\max - C_{1}/2n}}{K})$ to $\frac{C_{1}}{2n}+(\mu+1)(\frac{p_{\max - C_{1}/2n}}{K})$ \State $\hspace{ 4 mm} \textbf{do}$
		\State \hspace{ 4 mm} $\textbf{if } \left( q(j,i)<(p_{j}R^{1-\epsilon} \cdot q \cdot \gamma_{1} + (k-1)q)\right) $ $\textbf{then}$
		\State \hspace{ 6 mm} Calculate: $Da_{p}(j,i)= \min \left( 1, c_{i} \dfrac{p_{j} q_{i}}{p_{i} q_{ji}}\right)$
		\State \hspace{ 8 mm} $\textbf{if}$ $\left( Da_{p}(S,i)\leq C_{1}/\left( \beta\cdot C_{2}\right)\right)  $ and $DT(i)$ and $PF(i)$ 
		\State \hspace{ 8 mm} and $\left( Da_{p}(S,i)\leq C_{D}/\left( \beta\cdot C_{2}\right)\right)  $
		\State \hspace{ 8 mm} $\textbf{then}$
		\State  \hspace{ 10 mm} $i$ quits with color $p_{i}$
		\State  \hspace{ 10 mm} $S_{i}=S_{i-1}\cup{q_{i}}$
		\State  \hspace{ 10 mm} $p_{i}=2p_{i}$
		
		\State \hspace{ 8 mm} $\textbf{end if}$
		\State \hspace{ 4 mm} $\textbf{end if}$
		\State \hspace{ 4 mm} Output S
		\State \hspace{ 4 mm} $i$ quits with color $2\left( \frac{C_{1}}{2n}+(\mu+1)(\frac{p_{\max - C_{1}/2n}}{K})\right)$ 
		\State $\hspace{ 4 mm}\textbf{end for}$
		\State $\hspace{ 2 mm}\textbf{end for}$
		\State $\textbf{end for}$
	\end{algorithmic}
\end{algorithm}	The algorithm is based on the coloring method in \cite{Jurdzinski:2014}, which assigns probability/color to each node taking part in an implementation. Also, we consider that the message is successfully received when the sender receives an acknowledgement message to inform it. 

In Section \ref{sub2}, we prove the existence of a dense ball that the message is successfully received with high probability. In Section \ref{sub3}, we prove that there is a set of guards that is guarding the receiver of quasi-link $q_{i}$. We use guards in order to protect the receiver from interference of other quasi-links and to boost the signal. Also, we prove that there is a set of guards that is guarding the sender of quasi-link $q_{i}$. In the last case, we want to protect the sending of an acknowledgement message from its receiver. Then, we guarantee the successful transmission. 
\subsection{Overview of the Algorithm}\label{sub1}
First of all, we consider that Algorithm 1 determines a probability/color $p_{i}$ from the set of probability transmissions $\left\lbrace \frac{C_{1}}{2n}\mu(\frac{p_{\max - C_{1}/2n}}{K})| \mu \in [0,K] \text{ \& } K= \frac{2C_{1}\beta'}{\beta} \right\rbrace$ to each node in order to achieve transmissions of the messages in the network. Then, the number of colors is $O(\log^{*}\Delta)$, where $\Delta $ is the ratio between the maximum and the minimum power assignment as well $K= \frac{2C_{1}\beta'}{\beta}$ using the Theorem of \cite{APX} and the Theorem of \cite{HalldorssonT15}. Then, it holds the following: any $\beta/C_{1}$- feasible set can be partitioned into $\lceil \frac{2C_{1}\beta'}{\beta} \rceil$ subsets, each of which is $\beta'$- feasible.

In the next step of Algorithm 1, we use a restriction in order to schedule quasi-links whose the quasi-distance is at most $p_{j}R^{1-\epsilon} \cdot q \cdot \gamma_{1} + (k-1)q$, where $R$ is the ratio of the power assignments of quasi-links $i$ and $j$, $0<\epsilon<1$ and $k$ is a constant which indicates the disk $k$. Then, we control three states: 1) The restriction of affectance on quasi-link $q_{i}$ by other quasi-links $q_{j}$, 2) The verification if the messages are received with success (Algorithm 3) and 3) The verification if the network is density (Algorithm 3).

Algorithm 2 controls the sum of probabilities $p_{i}$ in a ball $B$ and waits acknowledgement transmissions. This algorithm returns true if the sum of probabilities $p_{i}$ in a ball $B$ is at most $C_{1}$, where  $0<C_{1}<1$. It means that the expected affectance from a set of quasi-links on a quasi-link is not large when all nodes wants to transmit concurrently using a probability transmission each of them.
\begin{algorithm}
	\caption{Controls if the messages are received with success}
	\begin{algorithmic}
		\State $\textbf{Procedure } DT(i):$
		\State $\textbf{while}$ $success\neq true$ $\textbf{do}$
		\State \hspace{ 2 mm} $\textbf{for}$ $c_{0}\log n$ rounds $\textbf{do}$ transmit with probability $p_{i}$
		\State \hspace{ 4 mm} wait for acknowledgment
		\State \hspace{ 4 mm} $\textbf{if}$ receives at least $c_{1}\log n$ messages $\textbf{then}$
		\State \hspace{ 6 mm} return True
		\State \hspace{ 4 mm} $\textbf{end if}$
		\State $\hspace{ 2 mm}\textbf{end for}$
	\end{algorithmic}
\end{algorithm}
Algorithm 3 controls the density. It means that the probability transmission is constant. More, applying the technique of \cite{Jurdzinski:2014}, Algorithm 3 blocks the sum of probabilities in a ball $B$ to overcome the constant $C_{1}$. This happens because of the positive results of both Algorithm 2 and Algorithm 3. After each successful transmission of nodes $i$, they disable and the sum of probabilities is reduced.
\begin{algorithm}
	\caption{Controls the density}
	\begin{algorithmic}
		\State $\textbf{Procedure } PF(i):$
		\State $\textbf{while}$ $success\neq true$ $\textbf{do}$
		\State \hspace{ 2 mm} $\textbf{for}$ $c_{2}\log n$ rounds $\textbf{do}$ transmit with probability $p_{i}\cdot c_{\epsilon}$
		\State \hspace{ 4 mm} wait for acknowledgment
		\State \hspace{ 4 mm} $\textbf{if}$ receives at least $c_{3}\log n$ messages $\textbf{then}$
		\State \hspace{ 6 mm} return True
		\State \hspace{ 4 mm} $\textbf{end if}$
		\State $\hspace{ 2 mm}\textbf{end for}$
	\end{algorithmic}
\end{algorithm}
More, there is the case that a probability transmission can overcome $p_{\max}$, it happens in sparse area. Then, Algorithm 1 is executed (the last line of algorithm) and the sum of probabilities is at least $C_{2}$.
\subsection{Without Guards}\label{sub2}
In Lemma \ref{lb}, we show that there is a dense ball in a ball $B$ using real conditions. In Lemma \ref{lpp}, we seek the probability of receiving a message using the case of acknowledgment transmissions.
\begin{lemma}\label{lb}
	Let $B \equiv \bigcup_{k>1}\left( B_{k}(i,Rq\gamma_{1}+(k-1)q)\backslash B_{k-1}(i,Rq\gamma_{1}+(k-2)q)\right)$. For each node $i\notin S$, we assume a ball $B$ , whenever $\sum_{j\in B} p_{j}\geq C_{1}/2$, there exists $w_{k}$, which is the center of $B(w_{k}, q/2)$ and is included in $B_{k}(i,Rq\gamma_{1}+ (k-1)q + \delta)\setminus B_{k-1}(i,Rq\gamma_{1}+ (k-2)q - \delta)$. Let $B_{1k}=B(w_{k}, q/2)$, $B_{2k}=B(w_{k}, hq/2)\setminus B(w_{k}, q/2)$ and $B_{3k}=B(w_{k}, q)\setminus (B_{1k}\cup B_{2k})$. Then,
	\begin{subequations}\label{ceq1}
		\begin{align}
		s \leq  \sum_{j\in B_{1k}(w_{k}, q/2)} p_{j} & \leq 1/2,   \text{ } \forall k>1\\
		\sum_{j\in B(x, q/2)} p_{j} & \leq \zeta^{m}s,   \text{ } \forall x\in B_{2k}\cup B_{3k}\\
		\sum_{u\in B(j, Rq\gamma_{1}+ (k-1)q -\delta_{1})} p_{u} & \geq \frac{C_{1}}{2 \chi(q,qR\gamma_{1}+(k-1)q)}, \vspace{2 mm} \text{ } \forall j\in B_{1k}
		\end{align}
	\end{subequations}	
\end{lemma}
\begin{proof}
	Let $i\notin S$ be a receiver node and $B \equiv \bigcup_{k>1}\left(  B_{k}(i,Rq\gamma_{1}+(k-1)q)\backslash B_{k-1}(i,Rq\gamma_{1}+(k-2)q)\right)$. We consider that $\sum_{j\in B} p_{j}\geq C_{1}/2$. Also, we assume that $B'$ ball with center $x'$ and radius $q/2$ is the largest mass of probability that $B'\subseteq B$. Then, $\sum_{j\in B'} p_{j} \geq \frac{C_{1}}{2 \chi(q/2,qR\gamma_{1}+(k-1)q)}$. If there is a node $x$ in a ball $B(x', Rq\gamma_{1}+ (k-1)q -\delta_{1})$ and satisfies the inequalities (\ref{ceq1}a), (\ref{ceq1}b) then is also satisfied the (\ref{ceq1}c) inequality (as it is proved in \cite{Jurdzinski:2014}).
	
	In addition, we consider a ball $B'_{0}\equiv B(x_{0},r_{0})\subseteq B$. The average probability mass of $B'_{0} \subseteq B$ is at least $s$. The average probability mass of $B'\subseteq B$ is at least $s$, where $B'_{0} \subseteq B'$ and $B'$ is the ball of radius $q/2$ with the highest probability mass. If the (\ref{ceq1}b) inequality is satisfied for a node $x_{0}$ then (\ref{ceq1}c) is also satisfied. If the (\ref{ceq1}b) inequality is not satisfied for a node $x_{0}$, there is a ball of radius $r_{0}$ and the probability mass is at least $\zeta^{m}s$ in distance at most $q/2$ from $x_{0}$. Then, there is a ball $B'_{1}\subseteq B$ with radius $q/\zeta$ and probability mass at least $s$, which are guaranteed by bounded growth property. Thus, a sequence of balls is created with probability mass at least $s$ that a ball $B'_{n}$ has radius $r_{n}=q/\zeta^{n}$ for $n>0$. The distance between $B'_{n}$ and $B'_{n+1}$ according to their centers is $d(B'_{n},B'_{n+1})\leq \frac{q}{2}(1+1/2\zeta+...+1/2\zeta^{n})\leq q$, which means that (\ref{ceq1}b) inequality is not satisfied. If (\ref{ceq1}b) is satisfied for some node $x_{n}$ and radius $r_{n}=q/2\zeta^{n}$ then the center of the ball $B'_{n}$ is in distance at most $q/2\sum_{n}1/\zeta^{n}\leq q\leq Rq\gamma_{1}+(k-1)q-\delta_{1}$. According to \cite{jurdzinski2005probabilistic}, we have an event success.      
\end{proof}
\begin{lemma}\label{lpp}
	Let $B(i,  Rq\gamma_{1}+ (k-1)q)$ be a ball satisfying the previous Lemma \ref{lb}, where $k>1$. Then, for every $j \in B(i, Rq\gamma_{1}+ (k-1)q))$, the probability of receiving a message $p(i)\geq \frac{s}{16}\cdot (\frac{1}{4})^{(2R\gamma_{1}+2k-4)^{\xi}\cdot s}$.	
\end{lemma}
\begin{proof}
	Seek the probability that the receiver $i$ of a quasi-link successfully gets the message that can be affected by other nodes $j \in B(i, Rq\gamma_{1}+ (k-1)q)\equiv B$. We use the acknowledgement transmission to inform its sender in odrer to guarantee the successful reception of message. Then, the probability of successfully receiving a message from a node is computed from the joint of the following probabilities:
	\begin{itemize}
		\item Event $E_{1}$. The message is transmitted from only one node, which belongs to the disk $D_{k}=B_{k}\setminus B_{k-1}$. Using Fact 4 of \cite{Jurdzinski:2014}: 		$\Pr[E_{1}] := \Pr[\text{exactly one node } j\in\{1,...,m\} \text{ transmits } \& j \in D_{k}] \geq \frac{s}{2}$, as $s = \sum_{j\in D_{k}} p_{j} \leq 1/2$. Thus, $\Pr[E_{1}]\geq \frac{s}{2}$.
		\item Event $E_{2}$. There are not sender-nodes in the ball $B_{k-1}$. Using Fact 5 of \cite{Jurdzinski:2014}:	There are $m$ nodes with probabilities $p_{j}\leq 1/2$ for each node $j\in \{1,...,m\}$ and $j \in B_{k-1}$ then $\Pr[\text{no node j transmits, }  j \in \{1,...,m\}] \geq \left( \frac{1}{4}\right)^{\sum_{j \in B_{k-1}} p_{j}}$. Also, we observe that the ball $B_{k-1}(i,Rq\gamma_{1}+(k-2)q)$ consists of $(2R\gamma_{1}+2k-4)^{\xi}$ balls with radius $q/2$ and $\xi$ be the dimension. Then, $\sum_{j \in B_{k-1}} p_{j}\leq (2R\gamma_{1}+2k-4)^{\xi}\cdot s$. Thus, $\Pr[E_{2}]\geq \left( \frac{1}{4}\right)^{(2R\gamma_{1}+2k-4)^{\xi}\cdot s}$.
		\item Event $E_{3}$. The sender $j \in D_{k}$ receives the acknowledgement message from its receiver, given that a message is transmitted from the sender-node $ j$. Let $Ack_{j}$ be the random variable, that is equal to 1 if the message is successfully received, otherwise 0. Thus, we compute $\Pr [Ack_{j}=1 | j \text{ transmits } \& j \in D_{k}]$. Using Markov inequality in the second inequality: $\Pr [Ack_{j}=0 | j \text{ transmits } \& \hspace{1 mm} j \in D_{k}] \leq \Pr[\sum_{j} Da_{p}(j,i)\cdot Ack_{j}\geq \frac{C_{1}}{\beta C_{2}}] \leq \frac{\beta C_{2}}{C_{1}}\mathrm{E} [\sum_{j} Da_{p}(j,i)\cdot Ack_{j}]= \frac{\beta C_{2}}{C_{1}}\sum_{j} a_{p}(j,i)\cdot q \leq \frac{1}{2}$. The transmission probability is\\ $q\leq \frac{1}{2} \beta^{2-\xi} (\frac{C_{2}}{C_{1}})^{3-\xi}$, for $S'=\left\lbrace i\in S: \sum_{j} a_{p}(j,i) \leq \frac{C_{1}}{C_{2}} O(\gamma_{1}^{\xi-2}) \right\rbrace $. Therefore, $\Pr[E_{3}]\geq \frac{1}{2}$.
		\item Event $E_{4}$. The WAFF by nodes from the union of disks $D_{k+1}\cup D_{k+2}\cup...\cup D_{K}$ on quasi-link $i$ is at most $\frac{C_{1}}{C_{2}}O(\gamma_{1}^{\xi-2})$. Let $S'=\left\lbrace j\in B_{K-r+1}\setminus B_{K-r}\equiv D_{k+1}\cup D_{k+2}\cup...\cup D_{K}\right\rbrace$ be a set of nodes. By Markov inequality, $\Pr [Da_{p}(S',i)\geq \frac{C_{1}}{C_{2}}O(\gamma_{1}^{\xi-2})] \leq \frac{C_{2}\cdot \mathbb{E}[Da_{p}(S',i)]}{C_{1} \cdot O(\gamma_{1}^{\xi-2})}$. The expected value is given by $\mathbb{E}[Da_{p}(S',i)]= \mathbb{E}\left[ \frac{\sum\limits_{j\in S'} a_{p}(j,i)p_{j}}{\sum\limits_{j\in S'} p_{j}}\right] \leq \frac{1}{C_{2}} \sum\limits_{j\in S'} a_{p}(j,i)\cdot \zeta^{\xi}\cdot s =\frac{1}{C_{2}} \sum\limits_{K-r\geq 2}\hspace{1 mm}\sum\limits_{j\in B_{K-r+1}\setminus B_{K-r}} a_{p}(j,i)\cdot \zeta^{\xi}\cdot s	\leq \frac{R\cdot q_{i}}{C_{2}\cdot q} \sum\limits_{K-r\geq 2} \frac{\mid B_{K-r+1}\mid- \mid B_{K-r}\mid}{R\gamma_{1}+k-2} \cdot \zeta^{\xi}\cdot s\leq \frac{R\cdot q_{i}}{C_{2}\cdot q} \sum\limits_{K-r\geq 2} \frac{2R\gamma_{1}+2(k-r)^{\xi-1}}{R\gamma_{1}+k-2} \cdot \zeta^{\xi}\cdot s\leq \frac{R\cdot q_{i}}{C_{2}\cdot q} \sum\limits_{K-r\geq 2} 2R\gamma_{1}+2(k-r)^{\xi-2} \cdot \zeta^{\xi}\cdot s$.
		
		Therefore, we choose $s\leq \frac{1}{4}\cdot \frac{q\cdot C_{2}^{2} O(\gamma_{1}^{\xi-2})}{\zeta^{\xi}\cdot \sum\limits_{K-r\geq 2} 2R\gamma_{1}+2(k-r)^{\xi-2}}$ in order to bound the expected value by $\frac{C_{2}\cdot }{2\cdot C_{1}}\cdot O(\gamma_{1}^{\xi-2})$. Thus, $\Pr[E_{4}]\geq \frac{1}{2}$.
		\item Event $E_{5}$. The WAFF by nodes from outside of the ball $B_{K}$ is at most $\frac{C_{1}}{C_{2}}O(\gamma_{1}^{\xi-2})$. Then, $\Pr[E_{5}]\geq \frac{1}{2}$. The proof is based on the division of the space that nodes are located outside the ball $B_{K}$. The idea of the proof is same as the previous event.
	\end{itemize}
	Therefore, the probability of receiving a message is $p(i)\geq \frac{s}{16}\cdot (\frac{1}{4})^{(2R\gamma_{1}+2k-4)^{\xi}\cdot s}$. 
\end{proof}
\begin{lemma}\label{lp}
	Let $B_{1k}(w_{k},q/2)\subseteq B(i,  Rq\gamma_{1}+ (k-1)q)$ be a ball satisfying the previous Lemma \ref{lb}, where $k>1$. Then, for every $j \in B_{1k}(w_{k},q/2)$, the probability of receiving a message $p(i)\geq \frac{s}{16}\cdot (\frac{1}{4})^{(2R\gamma_{1}+2k-4)^{m}s}$.	
\end{lemma}
\subsection{With Guards}\label{sub3}
In this part, we seek a set of nodes that is guarding a receiver/sender node from the affectance of other quasi-links in a metric decay space, which is bounded-growth. This space has bounded independence dimension and bounded doubling dimension, which is defined in \cite{bodlaender2014beyond}. Specifically, we use guards in order to protect the receiver from the interference of other quasi-links and to boost the signal. Then, we show that there is a set of guards which is guarding a sender of a quasi-link $i$. Thus, we protect the sending of an acknowledgement message of its receiver from various interference types, such as the simultaneous transmissions of quasi-links. 

First of all, we need the following definition in order to prove the Lemma \ref{lemmaG}. We define a $(\mu-\delta)/\lambda$-density-dominant node in order to provide an efficient measure with regard to the transmission probability of nodes and the number of active nodes. In particular, we assume two disjoint subsets of nodes, the set of receiver $\mathcal{R}$ and the set  $S\setminus \mathcal{R}$,  where $S$ is a set of nodes.  We take a node $b \in S\setminus \mathcal{R}$ and we create a ball around $b$. Then, we determine the relation of the probability of transmission of nodes of two above sets using the ratio $(\mu-\delta)/\lambda$. This means that the ball around $b$ either contains $(\mu-\delta)/\lambda$  more nodes $c\in B(b, d_{1})\cap (S\setminus \mathcal{R})$ than the nodes $e\in B(b, d_{1})\cap \mathcal{R}$ or contains $(\mu-\delta)/\lambda$ higher sum of probability of transmissions $\sum_{c\in B(b, d_{1})\cap (S\setminus \mathcal{R})}p_{c} $ than $\sum_{e\in B(b, d_{1})\cap \mathcal{R}} p_{e}$. In \cite{goussevskaia2009capacity}, the authors take into account only the number of nodes which are in the ball around $b$ and not the sum of probabilities of transmissions of nodes.
\begin{definition}
	Let $\mu$, $\lambda$ be positive constants, $\mu > \lambda$ and $0<\delta<1/2$. We consider two disjoint sets of nodes in a metric decay space $(\mathcal{V},d)$, the set of receivers $\mathcal{R}$ and the set $S\setminus \mathcal{R}$, where $S$ is a set of nodes. Then, a node $b\in S\setminus \mathcal{R}$ has a $(\mu-\delta)/\lambda$-density-dominant if every ball $B(b, d_{1})$ contains $(\mu-\delta)/\lambda$ more density of nodes $\sum\limits_{c\in B(b, d_{1})\cap (S\setminus \mathcal{R})} p_{c}$ than $\sum\limits_{e\in B(b, d_{1})\cap \mathcal{R}} p_{e}$ and is expressed as follows:	
	\begin{equation}
	\sum\limits_{c\in B(b, d_{1})\cap (S\setminus \mathcal{R})} p_{c} > \frac{(\mu-\delta)}{\lambda}\sum\limits_{e\in B(b, d_{1})\cap \mathcal{R}} p_{e}
	\end{equation}
\end{definition}

Moreover, we define a set of guards which is guarding a receiver/sender. Also, we assume that each guard has a probability of transmission. Then, we bound the sum of probabilities of transmissions of guards. The definition is used in the Lemma \ref{lemmaWG} and is expressed as follows:
\begin{definition}\label{defgg}
	Let $\mathcal{R}$ be a set of receiver nodes and let $i\in \mathcal{R}$. Let $S\setminus \mathcal{R}$ be a set of nodes and $\mathcal{R}\cap \left( S\setminus \mathcal{R}\right)  = \emptyset$. Also, we assume a subset of nodes $G \subseteq\left(  S\setminus \mathcal{R}\right) $, which is called as a set of guards on receiver node $i$ if for each $b\in\left(  S\setminus \mathcal{R}\right)\setminus G $ we have that $\sum_{g\in B_{k}(b,q)\cap G} p_{g}\leq s/2$ with $B_{k}(b,q)\cap G\neq \emptyset$ and $q$ is the quasi-distance of $b$ to $i$.
\end{definition}	
\begin{property}\label{prop1}
	For all	nodes $v \in V$, there is a set of guards $G_{v}\subset V$ of at most $D$ points that guards $v$: $\min_{w \in G_{v}} q(z, w) \leq q(z,v), \forall z \in V\setminus \left\lbrace v\right\rbrace$.
\end{property}			
\begin{obs}\label{obs1}
	A guard node can be active when there are $j=1,...J$ nodes in its range guarding with $\sum_{j} p_{j} \geq s/2$, $s>0$. 	
\end{obs}

\begin{figure}
	\centering
	\includegraphics[width=2.5in]{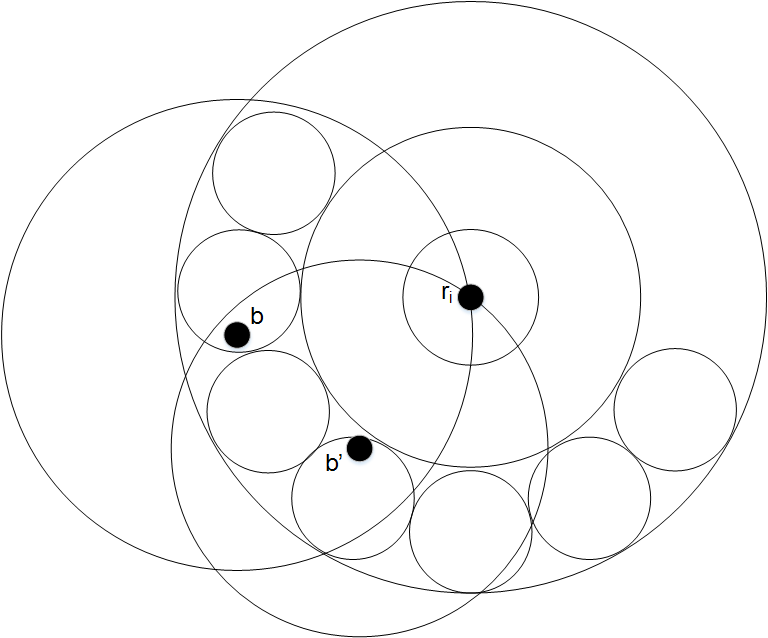}
	\caption{Concentric disks surrounded around the endpoint $r_{i}$ of quasi-link $i$. The node $r_{i}$ is guarded by a set of guards. The nodes $b,b' \in B_{k}(w_{k},q/2)$ and protect the node $r_{i}$.} 
	\label{fig:disksgr}
\end{figure}

In Lemma \ref{lemmaWG}, we prove that there is a set of guards in a ball $B$ and it is guarding the receiver of a quasi-link. The receiver is the center of $B$ which consists of concentric disks. Thus, we prove that there is a set of guards in a disk $D_{k}\left( \equiv B_{k}\setminus B_{k-1}\right)$. This set is activated in order to guard the receiver when there is a dense ball of nodes in $B$.
\begin{lemma}\label{lemmaWG}
	Let $B \equiv \bigcup_{k>1}\left( B_{k}(i,Rq\gamma_{1}+(k-1)q)\backslash B_{k-1}(i,Rq\gamma_{1}+(k-2)q)\right)$ be a ball. For each receiver $i$, whenever $\sum_{j\in B} p_{j}\geq C_{1}/2$ there exists always a set of guards $G_{k}\subseteq D_{k}\left( \equiv B_{k}\setminus B_{k-1}\right)$, with mass of probability of guards $g_{k}\in G_{k}$ is at least $C_{D_{k}}$ that is guarding the receiver node $i$. Then, $\sum_{g_{k}\in D_{k}} p_{g_{k}}\leq C_{D_{k}}$.
\end{lemma}
\begin{proof}
	Let $S$ be a set of nodes. We separate $S$ into $k$ concentric disks with center $i\notin S$ and $S_{k}\subseteq S$. We construct a ball $B$ around the node $i$ and the ball is divided into $k$ concentric disks, where $B=\cup_{k>1} B_{k}\setminus B_{k-1}$ and $S_{k}\equiv B_{k}\setminus B_{k-1}$. The initial ball $B_{1}= \emptyset$. We seek a set of guards which is guarding the receiver $i$. This set is denoted as $G$ and is defined by a disjoint of subsets $G_{1}\cup ...\cup G_{k}$, where $G_{k}\subseteq S_{k}$ and $k>1$. We assume that $G_{1}=S_{1}=\emptyset$.
	
	Let $i$ be the center of $B$ such that $\sum_{j\in B} p_{j}\geq C_{1}/2$ then there is a dense ball $B_{k}(w_{k},q/2)$ which is included in $B$ according to Lemma \ref{lb}. Thus, there are active nodes $j$, where $s \leq \sum_{j\in B_{k}(w_{k}, q/2)} p_{j} \leq 1/2$,  for each $k>1$ disk and $\sum_{j\in B_{k}(x,q/2)} p_{j}\leq \zeta^{\xi}\cdot s$ for each $x\in B_{2k}\cup B_{3k}$ when $B_{2k}=B(w_{k}, hq/2)\setminus B_{k}(w_{k}, q/2)$ and $B_{3k}=B_{k}(w_{k}, q)\setminus (B(w_{k}, q/2)\cup B_{2k})$.
	
	We add a node to an initially empty set $G_{k}$, where $k>1$, when the resulting set remains independent with respect to $i$, there is not affectance on the node $i$. Then, $G_{k}\subseteq B_{k}$ is guarding $i$ by the nodes $j\in D_{k}$. We assume that there is not a set $G_{k} \subseteq B_{k}$, which is guarding the receiver $i$,  with $\sum_{g\in G_{k}} p_{g}\leq C_{D_{k}}$. Then, $\sum_{g\in G_{k}} p_{g}\geq C_{D_{k}}$ and the nodes $g$ can affect the node $i$. There is a node $b\in B\setminus G_{k}$. We assume that $b\in D_{k}\setminus G_{k}$. By Definition \ref{defgg} and from our previous assumption, we have that $\sum\limits_{j\in B(b,d)\cap G_{k}} p_{j}\geq \frac{s}{2}$, where $d=d(b,i)$. Then,
	\begin{itemize}
		\item Let $b\in B_{k}(w_{k},q/2)$ as is represented in Figure \ref{fig:disksgr}. From Lemma \ref{lb}: $s \leq \sum_{j\in B_{k}(w_{k},q/2)} p_{j} \leq 1/2$. On one hand, the node $b\notin G_{k}$. On the other hand, we added a node $b'$ in a set $G_{k}\setminus B_{k}(b,d)$ and it is happend because the distance $d(b',i)$ is less than $d(b,i)$, where $d(b,i)\leq kq$ when $R\gamma_{1}=1$ and $k>1$. The node $b'$ was added before $b$. Note the Observation \ref{obs1} in which a guard node can be active when there are $b$ nodes in its range. Thus, $b\in B(b',d(b,i)) \Leftrightarrow d(b',b)\leq d(b,i)$. In the last inequality, we use the Property \ref{prop1}. We conclude to $b'\in B(b,d)$. It is a contradiction.
		\item Let $b\in B_{k}(x,q/2)$ for $x\in B_{2k}\cup B_{3k}$. From Lemma \ref{lb}: $\sum_{j\in B_{k}(x,q/2)} p_{j}\leq \zeta^{\xi}\cdot s$. This case is divided into three subcases: 1) $b\in G_{k-1}$ then there is $b' \in G_{k-1}\setminus B(b,d)$, 2) $b\in G_{k}$ then there is $b' \in G_{k}\setminus B(b,d)$, and 3) 1) $b\in G_{k+1}$ then there is $b' \in G_{k-1}\setminus B(b,d)$, where $d(b,i)\leq kq + \delta$  when $R\gamma_{1}=1$, $\delta\leq q/2$ and $k>1$. We conclude to $b'\in B(b,d)$ as in the first case, which is a contradiction.	
	\end{itemize}
\end{proof}
\begin{lemma}\label{lemmaG}
	Let $B \equiv \bigcup_{k>1}\left( B_{k}(i,Rq\gamma_{1}+(k-1)q)\backslash B_{k-1}(r,Rq\gamma_{1}+(k-2)q)\right)$ be a ball. Let $D_{k}\equiv(B_{k}\setminus B_{k-1})$ be a $k_{th}$ disk. Let $\mathcal{R}$ be the set of receivers and a constant $1\leq h\leq 2$. For each sender $s\in D_{k}$, there is a receiver node $r$ which transmits an acknowledgement message to its sender with probability $p_{r}\geq 1/2$ in a ball $B$. If  $\sum_{j\in B} p_{j}\geq C_{1}/2$ and $\sum\limits_{j\in B(s, h\cdot q/2)\cap (D_{k}\setminus \mathcal{R})} p_{j} > \frac{(\mu-\delta)}{\lambda}\sum\limits_{e\in B(s, h \cdot q/2)\cap\mathcal{R}} p_{e}$ then 
	\begin{itemize}
		\item There is always a guard $g \in G_{k}$, where $G_{k}\subseteq D_{k}$ and $k>1$, with mass of probability at least $C'_{D_{k}}$ and is given by $\sum_{g\in D_{k}} p_{g}\leq C'_{D_{k}}$.
		\item There is a density-dominant in a ball $B_{k}$ with  $\sum\limits_{j\in(B_{k}\setminus \mathcal{R})} p_{j} > \frac{(\mu-\delta)}{\lambda}\cdot C'_{B_{k}}\sum\limits_{e\in \mathcal{R}} p_{e}$.
	\end{itemize}	
\end{lemma}
\begin{proof}
	Let $S$ be a set of nodes. We separate $S$ into $k$ concentric disks with center $i\notin S$ and $S_{k}\subseteq S$. This means that we construct a ball $B$ around the node $i$. The idea of this proof is the same as Lemma \ref{lemmaWG}. Thus, this ball is divided into $k$ concentric disks, where $B=\cup_{k>1} B_{k}\setminus B_{k-1}$ and $S_{k}\equiv B_{k}\setminus B_{k-1}$. The initial ball $B_{1}= \emptyset$. 
	
	Firstly, we seek a set of guards which is guarding the sender $s$. This set is denoted as $G_{k}$ and is defined by a disjoint of subsets $G^{1}_{k}\cup G^{2}_{k}$. We assume that $s$ is the center of a ball with radius $q/2$ and $s\in D_{k}$. Also, we consider a larger annulus ball around $s$ with radius $a\cdot q/2$, where $1<a<2$. Therefore, $G^{1}_{k}$ is included in $B(s,q/2)$ and $G^{2}_{k}$ is included in $B(s,a\cdot q/2)$. Also, $G^{1}_{1}=G^{2}_{1}= \emptyset$.
	
	We add a node to an initially empty set $G_{k}$, when the resulting set remains independent with respect to $s$, there is not affectance on the node $s$. Then, $G_{k}\subseteq D_{k}$ is guarding $s$ by the nodes $j\in D_{k}$.
	
	We assume that there is not a set $G^{t}_{k}\subseteq B_{k}$ for $t=1,2$, which is guarding the receiver $s$,  with $\sum_{g\in G_{k}} p_{g}\leq C'_{D_{k}}$. This means that $\sum_{g\in G_{k}} p_{g}\geq C'_{D_{k}}$ and the nodes $g$ can affect the node $s$. The sender may not receive an acknowledgement message from its receiver. Thus, there is a node $b\in D_{k}\setminus G_{k}$ and the rest of this proof consists of two cases: 
	\begin{itemize}
		\item  Let $b\in B_{k}(s,q/2)$, $\sum_{j\in B} p_{j}\geq C_{1}/2$ and $\sum\limits_{j\in B(s, q/2)\cap (D_{k}\setminus \mathcal{R})} p_{j} > \frac{(\mu-\delta)}{\lambda}\sum\limits_{e\in B(s, q/2)\cap\mathcal{R}} p_{e}$. The node $s$ has a $\frac{(\mu-\delta)}{\lambda}$ density-dominant in $B(s, q/2)$. It means that there are more senders than receivers and therefore the affectance on $s$ is larger by other sender nodes in this region.
		
		Also, there is a node $b'$ in $G_{k}\setminus B_{k}(b,q(b,s))$ and this node was added before $b$ in this set because the quasi-distance $q(b',s)$ is less than $q(b,s)$, where $q(b,s)\leq q$. Note the Observation \ref{obs1} in which a guard node can be active when there are $b$ nodes in its range. Thus, $b\in B(b',q(b,s)) \Leftrightarrow q(b',b)\leq q(b,i)$. In the last inequality, we use the Property \ref{prop1}. Therefore, we conclude to $b'\in B(b,q(b,s))$. It is a contradiction.
		\item Let $b\in B_{k}(s,h\cdot q/2)\setminus B_{k}(s, q/2)$, $\sum\limits_{j\in B(s, h\cdot q/2)\cap (D_{k}\setminus \mathcal{R})} p_{j} > \frac{(\mu-\delta)}{\lambda}\sum\limits_{e\in B(s,h\cdot q/2)\cap\mathcal{R}} p_{e}$ for $1\leq h\leq2$ and $\sum_{j\in B} p_{j}\geq C_{1}/2$. The node $s$ has a $\frac{(\mu-\delta)}{\lambda}$ density-dominant in $B(s,h\cdot q/2)$. It means that there are more senders than receivers and therefore the affectance on $s$ is larger by other sender nodes in this region. Also, there is a node $b'$ in $G_{k}\setminus B_{k}(b,h \cdot q(b,s))$ and this node was added before $b$ in this set because the quasi-distance $q(b',s)$ is less than $q(b,s)$, where $q(b,s)\leq q$. We conclude to $b'\in B(b,q(b,s))$ as in the first case, which is a contradiction.
	\end{itemize}
	
	Secondly, we seek a density-dominant in a ball $B_{k}$ that   
	$\sum\limits_{j\in(B_{k}\setminus \mathcal{R})} p_{j} > \frac{(\mu-\delta)}{\lambda}\cdot C'_{B_{k}}\sum\limits_{e\in \mathcal{R}} p_{e}$, where there are receiver nodes less than nodes $j\in(B_{k}\setminus \mathcal{R})$. By using the first part of this lemma, we have a probability mass of guards at most a constant $C'_{D_{k}}$ in the disk $D_{k}$. Then, we can conclude to a constant $C'_{B_{k}}$ in a ball $B_{k}$ which consists of $k$ disks. That constant is the finite independence-dimension in a metric space.
\end{proof}
\section{Power Selection with Markov Chains - Online Algorithm}
In this section, we propose a new algorithm that is based on the procedures of SPAIDS algorithm of Section \ref{spaids1} as well as the acknowledgement messages of Section \ref{acknow1} in order to study another fundamental problem of wireless ad hoc networks, the broadcast problem. Also, we consider that the distances between nodes apply the symmetry property. Our aim is to find the expected successfull transmissions, the expected consumption of power and to compare the performance of the new proposed algorithm with the optimal algorithm in simple (metric) spaces. Thus, we study the online broadcast problem that each mobile device has a battery. Let $C_{B}$ be the maximum power that can be stored in the mobile device. We assume that there are $n$ users/nodes in the wireless network. Let $\mathcal{N}= \{1,...,n\}$ be the set of nodes. We consider that a node is activated in case that it receives packets (or else appears in online state) at each time step $t \in T$. There is unknown distribution of nodes in our network.

Then, an algorithm is constructed by a power function $p$ and an affectance function $F$. The power function $p:\mathbb{R}_{+}^{n}\rightarrow\mathbb{R}_{+}^{n}$ gives a vector of powers consumption of a sender-node to send a message to the online node-receivers in our network. In particular, $p_{j}(x)$ reflects the power of consumption of the sender when user $j$ receives the message, where $x$ is a vector that depends on the distance of nodes and the power of other nodes which simultaneously transmit a message. The affectance function $F$ chooses a subset $S\subseteq F$, where $S:= \left\lbrace \text{users/nodes } j: a_{p}(j,i)=\frac{p_{j}\cdot l_{i}^{a}}{p_{i}\cdot d(j,i)^{a}}\leq 1/\beta \right\rbrace$. For simplicity, in this case of affectance, the distances apply the symmetry property since the nodes in this network are in a metric space and not in a decay environment. We denote the following property as "battery-feasibility".	
\begin{property}
	The usage of mobile-phone is restricted, the maximum power that can be used is $C_{B}$ such that $\sum_{i} p_{i} \leq C_{B}$
\end{property}
The goal is to take the maximum subset of users that get the message from a sender node with a battery storage (capacity) at most  $C_{B}$ when the constraint of affectance is satisfied. Also, we need an optimal solution in order to compare the solution of OAMS algortihm. This means that we compare the performance of OAMS to the optimal. Full proofs and additional properties of this section are provided in the Appendix G.

Consider that a sender transmits with power $p_{1}=\frac{C_{B}}{n}$ to $r$ receivers. Using power $p_{i}$ then the sender wants to transmit the message to $\frac{rp_{i}n}{C_{B}}$ receivers. Thus, the algorithm is separated into three cases according to the number of receivers who can successfully receive the message. Also, the algorithm waits acknowledgment messages. In the first case, if the number of receivers is at least $r(1+\epsilon/2)$, then the algorithm updates the state $i$ to $i-1$ and the power $p_{i}$ to $p_{i-1}$. In the second case, if the number of receivers is at most $r(1-\epsilon/2)$. Then, the algorithm updates the state $i$ to $i+1$ and the power $p_{i}$ to $p_{i+1}$. Otherwise, the state $i$ and the power $p_{i}$ remain the same. 

Also, we use an index $s$ in order to choose an ideal power state with a probability $q$ that the message is successfully received. It means that the next conditions are satisfied: $p_{s-1} <\frac{C_{B}}{n q_{s-1}}$ and $p_{s}\geq \frac{C_{B}}{n q_{s}}$. Additional condition to determine the ideal power is the restriction of affectance. Thus, the definition of an ideal power state is given as follows:
\begin{definition}
	Let $p_{i}$	be the transmission power. The power $p_{i}$ is defined as ideal when 1) the condition of affectance is satisfied $a_{p}(S,i)\leq1/\beta$, 2) the probability that a receiver get the message with power $p_{i}$: $q_{i}> \frac{(1-\epsilon)C_{B}}{p_{i}n}$; and 3) the inequality of indexes $i\leq s$. Therefore, the set of ideal powers is defined as $PID=\left\lbrace p_{i}: i\leq s \text{ \& } a_{p}(S,i)\leq1/\beta \text{ \& } q_{i}> \frac{(1-\epsilon)C_{B}}{p_{i}n} \right\rbrace$. 
\end{definition}
\begin{algorithm}
	\caption{Online Algorithm in a Metric  Space (OAMS)}
	\begin{algorithmic}
		\State $\textbf{Initialization:}$ We consider a set of links $l_{1},...l_{n}$ and a set of power transmissions $p_{1}=C_{B}/n,p_{2}= \lambda\cdot C_{B}/n,...,p_{i}=\lambda\cdot p_{i-1}$ and $C_{B}/a\cdot s < p_{k} <C_{B}/s$ for $1\leq i< k$. Let $0<\gamma\leq 1$ and $\epsilon=0.1$. The initial set of feasible links $S=\emptyset$ and the initial state of power is $p_{i}=p_{k}$.
		\State $\textbf{Procedure SP(i):}$
		\State  $\textbf{for }$ $\frac{rp_{i}n}{C_{B}}$ receivers $\textbf{ do}$ transmit with power $p_{i}$
		\State \hspace{ 2 mm} $\textbf{if } a_{p}(S,i)\leq \gamma \textbf{ then}$ 
		\State \hspace{ 2 mm} wait for acknowledgment
		\State \hspace{ 4 mm} $\textbf{if}$ at least $r(1+\epsilon/2)$ receivers successfully get m
		\State \hspace{ 4 mm} $\textbf{then}$
		\State  \hspace{ 6 mm} update $i$ to $i-1$ and quits with power $p_{i}=p_{i-1}$
		\State \hspace{ 4 mm} $\textbf{else if}$ at most $r(1-\epsilon/2)$ receivers successfully get m
		\State \hspace{ 4 mm} $\textbf{then}$
		\State  \hspace{ 6 mm} update $i$ to $i+1$ and quits with power $p_{i}=p_{i+1}$
		\State \hspace{ 4 mm} $\textbf{else}$
		\State  \hspace{ 4 mm} quits power $p_{i}$
		\State  \hspace{ 4 mm} $S=S\cup \{l_{i}\}$
		\State \hspace{ 4 mm} $\textbf{end if}$		
		\State \hspace{ 2 mm} $\textbf{end if}$
		\State \hspace{ 2 mm} Output S
	\end{algorithmic}
\end{algorithm}

In Lemma \ref{eq:W1} and Lemma \ref{eq:W2}, we determine the process of power in order to find, with high probability, the ideal power state.
\begin{lemma}\label{eq:W1}
	Consider power $p=p_{i}$ and $p^{*} \in PID$. If $p<p^{*}$, $\forall p^{*}$ then the algorithm updates the power state from $p_{i}$ to $p_{i+1}$ with $\Pr[\#\text{receivers} \leq (1-\epsilon/2)r]\geq 1 -(\frac{1}{\epsilon^{2}r} - \frac{1}{\epsilon r})$. 
\end{lemma}
\begin{proof}
	Let $p=p_{i}$ be the power. The expected number of receivers who successfully get the message with $p$ is given as follows:$\mathbb{E}[\#\text{successes with p}]= \frac{r n p}{C_{B}}\cdot \frac{q_{p}}{2}$. By Chebyshev's inequality, $\Pr[\#\text{receivers} \geq (1-\epsilon/2)r] < \frac{r n p q_{p} }{2(1-\epsilon)rC_{B}}$. Then, the probability that a message is successfully received using power $p$ is given by $q_{p}> \frac{2(1-\epsilon)C_{B}}{n p}$.
	
	In case that the probability is small ($q_{p}\leq \frac{2(1-\epsilon)C_{B}}{n p}$), then the expected number of receivers is at most $(1-\epsilon)r$. By Chebyshev's Inequality, $\Pr[\#\text{receivers} \geq (1-\epsilon/2)r] \leq \frac{(1-\epsilon)r }{(\epsilon r)^{2}}$.
\end{proof}
\begin{lemma}\label{eq:W2}
	Consider power $p=p_{i}$ and $p^{*} \in PID$. If $p>p^{*}$, $\forall p^{*}$ then the algorithm updates the power state from $p_{i}$ to $p_{i-1}$ with $\Pr[\#\text{receivers} \geq (1+\epsilon/2)r]\geq 1 -(\frac{1}{\epsilon^{2}r} - \frac{1}{\epsilon r})$. 
\end{lemma}
\begin{proof}
	We examine the case that the probability $q_{p}$ is more than $q_{p_{s}}>\frac{2(1-\epsilon)C_{B}}{n p_{s}}$. Given $p>p^{*}$ and $p_{s}\geq p^{*}$, we observe that it holds $q_{p_{s}}>\frac{2(1-\epsilon)C_{B}}{n p_{s}}\geq \frac{2(1-\epsilon)\lambda C_{B}}{n p}> \frac{2(1-\epsilon) C_{B}}{n p}$. Then, the expected number of receivers is at least $(1-\epsilon)r$. By Chebyshev's Inequality, we have
	$\Pr[\#\text{receivers} \leq (1+\epsilon/2)r] \leq \frac{(1-\epsilon)r }{(\epsilon r)^{2}}$.	
\end{proof}

The effectiveness of the algorithm "OAMS" can be analyzed through the study of Markov chains. Let $Ch_{1}$ be a Markov chain. Thus, the algorithm "OAMS" can be represented as $Ch_{1}$. Specifically, the chain includes states $1\leq i \leq k$. Also, there are three cases of transitions from state $i$ to another state with a transition probability, i.e. $i$ to $i-1$, $i$ to $i+1$; and $i$ to itself. However, we need a simpler algorithm in order to compare our algorithm. Thus, the simple algorithm can be interpreted as a Markov Chain which is defined as $Ch_{2}$. The chain $Ch_{2}$ includes states $i> 0$ and there are two cases of transitions from state $i$ to another state with a transition probability, i.e. $i$ to $i-1$ with high transition probability $1 -(\frac{1}{\epsilon^{2}r} - \frac{1}{\epsilon r})$ and $i$ to $i+1$ with low transition probability $\frac{1}{\epsilon^{2}r} - \frac{1}{\epsilon r}$. 
\begin{lemma}
	Let $Ch_{1}$ and $Ch_{2}$ be two Markov Chains. The $Ch_{1}$ and $Ch_{2}$ include states $i\in \left\lbrace 1,...,k \right\rbrace$ and $i> 0$, respectively. In $Ch_{2}$, the state $0$ is an absorptive state. The $Ch_{1}$ is coupled with the $Ch_{2}$ as follows: the pair $(i_{t},j_{t})$ reflects the current states $i_{t}$ of $Ch_{1}$ and $j_{t}$ of $Ch_{2}$, for each time step $t\geq 0$. Let $i_{0}=i$ and $j_{0}=\min_{s\in PID} \left\lbrace |i-s|,|i-s+1|\right\rbrace$ be the beginning states of $Ch_{1}$ and $Ch_{2}$, respectively. If $i_{m}\notin PID$ in the $Ch_{1}$ for all $i_{0}\leq i_{m}\leq i_{t-1}$, then $j_{t}\geq \min_{s\in PID} \left\lbrace |i_{t}-s|,|i_{t}-s+1|\right\rbrace$ for the pair $(i_{t},j_{t})$.   	
\end{lemma}
\begin{proof}
	We consider that the current states of $Ch_{1}$ and $Ch_{2}$ is $(i,j)$. Then, we examine the next cases: Firstly, if the state $i$ is in $PID$ then we have an arbitrary pair $(i,j)$. Secondly, if the state $i$ is not in $PID$ then there is a state $u$ which is closer to PID than $i$. We assume that $u$ is the closest state of $i$ and $l$ is in PID. We have the following subcases: 1) If $i<s$ then $u=i+1$ for each $s$. 2) If $i>s$ then $u=i-1$ for each $s$. 
	
	In $Ch_{1}$, the transition $i$ to $u$ holds with $a_{i}$ probability; and the transition $i$ to $i$ holds with $\zeta_{i}$ probability. In the pair $(Ch_{1},Ch_{2})$, the update of $(i,j)$ to $(u,j-1)$ holds with transition probability $\frac{1}{\epsilon^{2}r} - \frac{1}{\epsilon r}$. The update of $i$ to $u$ holds with transition probability $a_{i}-\frac{1}{\epsilon^{2}r} - \frac{1}{\epsilon r}$, the update of $i$ to $i$ holds with $s_{i}$. The update of $i$ to $2i-u$ holds with $1 - a_{i}-\zeta_{i}$. The update of $j$ to $j+1$ holds  with  $1-\frac{1}{\epsilon^{2}r} - \frac{1}{\epsilon r}$. Therefore, if there are two consecutive state pairs that $i_{t-1}\notin PID$ then $j_{t}- j_{t-1}\geq \min_{s\in PID} \left\lbrace |i_{t}-s|,|i_{t}-s+1|\right\rbrace - \min_{s\in PID} \left\lbrace |i_{t-1}-s|,|i_{t-1}-s+1|\right\rbrace$.
\end{proof}

\begin{lemma}\label{lemmanumber}
	The algorithm OAMS makes at most $\frac{r n a}{\xi(a-1+b)}$ successful transmissions in expectation and consumes at most $\frac{r(1+\frac{\epsilon}{2})a C_{B}}{\xi (a-1+b)}$ of the maximum power in expectation, which is stored in the mobile device, where $b= \left( \frac{1}{\epsilon^{2}r} - \frac{2}{\epsilon r}\right) (1- a^{2})$
\end{lemma}
\begin{proof}
	We consider that the $Ch_{1}$ is coupled with the $Ch_{2}$. The pair $(i_{t},j_{t})$ reflects the current states $i_{t}$ and $j_{t}$ of the first and the second Markov chain, respectively, for each time step $t\geq 0$. More, the starting state of the $Ch_{1}$ is $i_{0}=i$ and the starting state of $Ch_{2}$ is $j_{0}=\min_{s\in PID} \left\lbrace |i-s|,|i-s+1|\right\rbrace$. Thus, the rest of this proof consists of four cases:
	\begin{itemize}
		\item If state $i=j_{0} + s$ then the OAMS algorithm successfully sends the message to $\frac{rp_{i}n}{C_{B}}$ receivers in expectation which is equal to $\frac{rp_{s}na^{j_{0}}}{C_{B}}$. Also, the algorithm consumes at most $r(1+\epsilon/2)p_{i}$ in expectation and is equal to $r(1+\epsilon/2)p_{s}a^{j_{0}}$.
		\item If state $i=s-j_{0}$. The OAMS successfully sends the message to $\frac{rp_{i}n}{C_{B}}$ receivers in expectation which is equal to $\frac{rp_{s}na^{-j_{0}}}{C_{B}}$. Also, the algorithm consumes at most $r(1+\epsilon/2)p_{i}$ in expectation and is equal to $r(1+\epsilon/2)p_{s}a^{-j_{0}}$. \item If state $i=s-1-j_{0}$, the behaviour of the algorithm is the same.
		\item If state $i=s-1+j_{0}$, the behaviour of the algorithm is the same.
	\end{itemize}

	Therefore, the expected number of receivers who successfully get the message with $p_{i}$ is given by $\mathbb{E}[\#\text{successes with } p_{i}] \leq \sum\limits_{t=0}^{T-1}\frac{r n p_{i}}{C_{B}} =\sum\limits_{t=0}^{T-1}\frac{r n p_{s} a^{j_{t}}}{C_{B}} = \frac{r n p_{s} a\cdot(a^{j}-1)}{(a-1+b)C_{B}}$, where $b=\frac{1}{\epsilon^{2}r} - \frac{2}{\epsilon r} -\frac{a^{2}}{\epsilon^{2}r} + \frac{2a^{2}}{\epsilon r}= \left( \frac{1}{\epsilon^{2}r} - \frac{2}{\epsilon r}\right) (1- a^{2})$. The result of the equality is given by the linear recurrence. Also, we consider that the algorithm starts in the $k$-th state that the initial power is $p_{i}=p_{k}$. Then, $p_{k}= a^{k-1}p_{1}= a^{k-1} \frac{C_{B}}{n}$. Also, $p_{k}\leq \frac{C_{B}}{\xi}$, it means that $a^{k-1}\leq \frac{n}{\xi}$. 
	
	Therefore, $\mathbb{E}[\#\text{successes with } p_{i}] \leq \frac{r n p_{s} a\cdot(a^{k-s}-1)}{(a-1+b)C_{B}} < \frac{r n p_{s} a^{2-s} n}{(a-1+b)C_{B}\xi} = \frac{r n a^{s-1} C_{B} a^{2-s} n}{(a-1+b)C_{B}n \xi} = \frac{r n a}{\xi(a-1+b)}$.
	
The expected consumption of power is given by $\mathbb{E}[\text{power consumption}]\leq \sum\limits_{t=0}^{T-1} r(1+\frac{\epsilon}{2})p_{i} =\sum\limits_{t=0}^{T-1}r(1+\frac{\epsilon}{2})p_{s} a^{j_{t}} = \frac{r(1+\frac{\epsilon}{2})p_{s} a\cdot(a^{j}-1)}{a-1+b}$. Using the inequalities $p_{k}\leq \frac{C_{B}}{\xi}$ and $a^{k-1}\leq \frac{n}{\xi}$, we have the following expectation: $\mathbb{E}[\text{power consumption}]\leq \frac{r(1+\frac{\epsilon}{2})a C_{B}}{\xi (a-1+b)} $, where $b= \left( \frac{1}{\epsilon^{2}r} - \frac{2}{\epsilon r}\right) (1- a^{2})$.	
\end{proof}
\begin{theorem}\label{theoremoam}
	The OAMS algorithm is constant-competitive, in which the number of receivers who can successfully get the message is more than $C\cdot C_{B}/p_{s}$, where $C$ is a constant, with high probability.
\end{theorem}
\begin{proof}
	Seek the probability that a number of receivers can successfully receive the message from a sender node and we prove that it is more than $(C\cdot C_{B})/p_{s}$, where $C$ is a constant, with high probability. The proof is comprised of two cases:
	
	In the first case, we consider that $\frac{C_{B}}{p_{s}}\geq K$, where $K$ is a nonnegative constant. This bound determines that the OAMS runs the procedure SP(i) many times until the ratio $C\cdot\frac{C_{B}}{p_{s}}$ goes beyond the number of receivers who can successfully get the message. The probability of this case is calculated by the union of five following events. Also, the finding of the ratio bound is preferred when the five events are not happened. 
	\begin{itemize}
		\item Event $E_{1}$. The number of receivers that is greater than $\frac{20 r n a}{\xi(a-1+b)}$ using Lemma \ref{lemmanumber}, where $b= \left( \frac{1}{\epsilon^{2}r} - \frac{2}{\epsilon r}\right) (1- a^{2})$. In this event, we calculate the number of receivers from the initialization of OAMS algorithm. Using Markov's inequality, $\Pr [E_{1}] \leq 1/20$.
		\item Event $E_{2}$. The maximum quantity of power which is consumed and is larger than $\frac{20 r(1+\frac{\epsilon}{2})a C_{B}}{\xi (a-1+b)}$, where $b= \left( \frac{1}{\epsilon^{2}r} - \frac{2}{\epsilon r}\right) (1- a^{2})$. Then, $\Pr [E_{2}] \leq 1/20$ by Markov's inequality.
		\item Events $E_{3}$ and $E_{4}$. We take into account the ideal power state. Then, $\Pr [E_{3}] \leq 1/20$ and $\Pr [E_{4}] \leq 1/20$ by Markov's inequality.
		\item Event $E_{5}$. The number of receivers which is at most $r(1-2\epsilon)$. Then, $\Pr [E_{5}] \leq 1/20$.
	\end{itemize}
	Therefore, the probability that the number of receivers can successfully receive the message from a sender node is more than $C\cdot C_{B}/p_{s}$ with high probability. This means that this probability is equal to $1-\sum\limits_{\ell=1}^{5}\Pr [E_{\ell}]\geq 3/4$.
	
	In the second case, we assume that $\frac{C_{B}}{p_{s}}\leq K$, where $K>0$ is a constant. Then, the total power, which is stored in the mobile device, is consumed. Thus, the number of receivers who can successfully get the message is more than $\frac{C_{B}}{K \cdot p_{s}}$, with high probability.
\end{proof}
\section{Conclusion}
In this paper, we proposed the first randomized scheduling and power selection algorithm in a decay space (SPAIDS), where each node has its own power assignment and the distances between nodes are not symmetrical. We studied a realistic wireless network, which is beyond the geometry. Our model was based on the SINR model. The SPAIDS achieves $O(\log\log\Delta \log n)$ rounds whp, where $\Delta$ is the ratio of the maximum and minimum power assignments and $n$ is the number of nodes. Our idea was to assign probability/color to each node using the coloring method of \cite{Jurdzinski:2014} and taking into account sets of guards to ensure the successful transmission of messages through a complex space. Also, we proposed an online algorithm in a metric space (OAMS) that solves the broadcast problem where nodes are activated to receive packets. We computed the maximum subset of nodes that received the message from a sender node using a battery capacity. We showed that the OAMS is a constant-competitive algorithm.

In the future, the reduction of the time complexity can be investigated through alternative methods in a more general network using additional strict conditions (Rayleigh and Ricing fading) as well as acknowledgement messages. More, the online version of the proposed algorithm can be studied in a decay space.	

\end{document}